\documentclass[12pt]{article}
\usepackage{fullpage} % Sets margins to 1 inch. Remove when using different template

\usepackage{amsmath, amsfonts, amsthm}

\usepackage{graphicx}
\usepackage{xcolor}
\usepackage{caption}
\usepackage{subcaption} % Used to give captions to subfigures 

\usepackage{algorithm}
\usepackage{algpseudocode}
\usepackage{enumitem}
\usepackage{xspace} % Used in D2FA command to fix space after command

\usepackage[nocompress]{cite}

\usepackage{hyperref}
\usepackage{cleveref}

\newtheorem{theorem}{Theorem}
\newtheorem{lemma}{Lemma}
\theoremstyle{definition}

\newtheorem{corollary}{Corollary}[lemma]

\newcommand{\DdFA}{\ensuremath{\mathrm{D^2FA}}\xspace}
\newcommand{\CursiveDdFA}{$D^2\!\text{\itshape FA}$\xspace}

\usepackage{enumitem}

\newcommand{\dasheddivider}{
  \par\vspace{10pt}\noindent
  \cleaders\hbox to 6pt{\hss -\hss}\hfill\vbox{}
  \par\vspace{10pt}
}

\newcommand{\vv}[2]{\ensuremath{v_{#1,#2}}} % overskriver kommandoen for et "v" over et bogstav
\newcommand{\vin}[1]{\vv{#1}{\text{in}}}
\newcommand{\vout}[1]{\vv{#1}{\text{out}}}
\newcommand{\tstate}{sink-state\xspace}
\newcommand{\tstates}{sink-states\xspace}
\newcommand{\vijstate}{set-state\xspace}
\newcommand{\vijstates}{set-states\xspace}
\newcommand{\inoutstate}{boundary-state\xspace}
\newcommand{\inoutstates}{boundary-states\xspace}
\title{Hardness and Approximation for Compressing Deterministic Finite Automata}
\author{Philip Bille\\\texttt{phbi@dtu.dk} \and Inge Li Gørtz\\\texttt{inge@dtu.dk} \and Christoffer Krogh\\\texttt{chrkro@dtu.dk} \and Christian Møller Mikkelstrup\\\texttt{cmomi@dtu.dk}}
\date{}

\begin{document}
\maketitle

\begin{abstract}
    The \emph{delayed deterministic finite automaton} ($\text{D}^2\text{FA}$) is a powerful tool for compressing the prohibitively large \emph{deterministic finite automata} (DFAs) used in modern signature-based deep packet inspection systems. A $\text{D}^2\text{FA}$ augments a DFA with unlabeled default transitions to remove transitions that are similar. This technique works well in practice and has inspired numerous subsequent studies. However, existing approaches do not guarantee a $\text{D}^2\text{FA}$ of optimal size as they restrict the default transitions to form a forest.

    In this paper, we consider the problem of finding an optimal $\text{D}^2\text{FA}$. We make the key observation that cycles of default transitions can be valid and produce a smaller $\text{D}^2\text{FA}$, provided that for every character in the alphabet, at least one state in the cycle has an outgoing transition labeled by that character. This observation raises fundamental questions concerning the computational complexity of finding an optimal $\text{D}^2\text{FA}$ and the approximation guarantee of the standard (forest-restricted) approach.
    
    We answer both questions by showing that it is NP-hard to find an optimal $\text{D}^2\text{FA}$ and by proving that the standard approach provides a $3/2$-approximation to the optimum. To establish NP-hardness, we construct a reduction from set cover using a DFA with gadgets where valid default transition choices correspond to selecting sets. To prove our approximation ratio, we compare the cycles in an optimal solution to a variant where the default transitions form a forest. Finally, we demonstrate that this $3/2$ ratio is tight via a constructed family of DFAs.
\end{abstract}
\newpage
\section{Introduction}
Signature-based deep packet inspection is a fundamental technique in modern network-security applications for identifying malicious content, sensitive information, and security-policy violations in network traffic. This technique relies on a collection of predefined patterns, represented as regular expressions, that are matched against an input (or data stream). The natural approach to this pattern matching is to use the \textit{deterministic finite automaton} (DFA) to represent the patterns and then simulate it over the input. This DFA-based approach to matching is used in the Snort~\cite{roesch1999a}, Zeek~\cite{paxson1999a}, and Suricata~\cite{suricata} intrusion detection systems (see also \href{https://snort.org/}{\texttt{snort.org}}, \href{https://zeek.org/}{\texttt{zeek.org}}, and \href{https://suricata.io/}{\texttt{suricata.io}}). A DFA consists of states and labeled transitions between states. Every state has an outgoing transition for each character in an alphabet.

Unfortunately, DFAs for the collections of regular expressions used in modern intrusion detection systems are prohibitively large and not feasible for practical implementation. To overcome the space issue, significant work has been done on compressing DFAs~\cite{6211894,ANTONELLO2015103,4215710,becchi2007a,1477950,becchi2013a,1635952,1452339,5643205,gong2023a,1460879,IntroOfD2FA,kumar2006a,liu2014a,liu2017a,meiners2010a,patel2014a,5935021,1354682,FastDFACompression}. One elegant and powerful approach is the \emph{delayed} deterministic finite automaton (\DdFA) introduced by Kumar, Dharmapurikar, Yu, Crowley, and Turner~\cite{IntroOfD2FA} and applied in many subsequent works~\cite{becchi2007a,becchi2013a,gong2023a,kumar2006a,liu2014a,liu2017a,matousek2018a,matousek2018b,meiners2010a,patel2014a,FastDFACompression}. Although this approach performs remarkably well in practice, we still know very little about its theoretical properties and performance guarantees.
% This technique leverages the observation that many states in the DFAs for real-world regular expression collections have similar sets of outgoing transitions, in many cases accounting for more than 95\% of transitions~\cite{IntroOfD2FA}.

The key idea behind the \DdFA is to leverage the observation that many states in the DFAs for real-world regular expression collections have similar sets of outgoing transitions~\cite{IntroOfD2FA}. Specifically, a \DdFA is a DFA where each state can additionally have an outgoing unlabeled transition, called a \emph{default transition}. A default transition from state $u$ to state $v$ can reduce the number of transitions in a \DdFA if $u$ and $v$ have a set of transitions with matching labels and target states. In this case, we can remove these transitions from $u$ and look for them at $v$ after following the default transition. Finding a transition with the desired label may require following a chain of default transitions, thus incurring a delay (as the name suggests). Two \DdFA{s} over the same set of states are equivalent if, starting from any given state and reading any character, both automata reach the same next state. Likewise, a \DdFA and a DFA are equivalent if they satisfy this same condition. The size of a \DdFA (or DFA) is the number of transitions and default transitions (where a DFA has zero default transitions). Among all equivalent \DdFA{s}, those of minimum size are considered optimal. The \emph{\CursiveDdFA compression problem} asks for an optimal (i.e. minimum-size) \DdFA equivalent to a given input DFA.

% To solve the \DdFA compression problem, all previous works determine the \DdFA using a variant of the following approach: construct a weighted complete graph over the DFA states—commonly called the \emph{space reduction graph} (SRG)—where each edge represents a candidate default transition weighted by the number of removable transitions. Then, compute a maximum spanning forest on the SRG. We use \emph{MSF compression} to denote any \DdFA that maximizes size reduction under the constraint that default transitions form a forest. However, despite the many practical studies of \DdFA{s}, there are, to the best of our knowledge, no theoretical studies of the relationship between MSF compression and an optimal unrestricted \DdFA.

To address the \DdFA compression problem, previous work uses a weighted complete graph over the DFA states, known as the \emph{space reduction graph} (SRG). In this graph, each edge represents a candidate default transition and is weighted by the number of transitions that can be removed after adding this candidate.
% a possible default transition weighted by the number of original transitions it removes. 
Existing algorithms compute a maximum spanning forest on the SRG and use the corresponding default transitions to compress the input DFA. We refer to any \DdFA that maximizes size reduction under this forest constraint as an \emph{MSF compression}.

% When building a \DdFA, we must ensure that following a chain of default transitions eventually leads to a transition labeled by the desired character. To the best of our knowledge, every previous approach ensures this by avoiding cycles of default transitions. We use \emph{MSF compression} to denote any \DdFA that maximizes size reduction under the constraint that default transitions form a forest. One standard way to obtain an MSF compression is to build a maximum spanning forest on a weighted complete graph over the DFA states. In this complete graph, each edge represents a candidate default transition between the two states and is weighted by the number of transitions that can be removed after adding this candidate. However, despite the many practical studies of the \DdFA{s}, there are, to the best of our knowledge, no theoretical studies of the relationship between the MSF compression and an optimal unrestricted \DdFA. 

We highlight the important observation that the default transitions of a \DdFA need not be restricted to form a forest. In fact, cycles of default transitions are completely valid as long as, for every character in the alphabet, at least one state in each cycle has an outgoing transition labeled by that character. 
% Importantly, cycles of default transitions are valid as long as, for every character in the alphabet, at least one state in each cycle has an outgoing transition labeled by that character. 
As demonstrated in \Cref{fig:allowingcyclesinDFA}, lifting the forest restriction allows us to construct an optimal \DdFA that is smaller than the MSF compression. 
% As seen in \Cref{fig:allowingcyclesinDFA}, we can use this property to obtain an optimal \DdFA that is smaller than the MSF compression. 
This observation gives rise to two natural questions:
\begin{itemize}
    \item What is the complexity of computing an optimal \DdFA?
    \item How far from the optimal \DdFA is the MSF compression?
\end{itemize}

\subparagraph*{Our contributions.}
In this paper, we answer both of these questions. We establish that computing an optimal \DdFA is NP-hard via a polynomial-time reduction from the set cover problem, and we prove that the MSF compression approach is a $3/2$-approximation to the problem. Formally, our main results are as follows: 

\begin{theorem}\label{thm:set_cover_reduction}
    The \CursiveDdFA compression problem is NP-hard.
\end{theorem}

\begin{theorem}\label{thm:apxRatio}
    The MSF compression is a $3/2$-approximation for the \CursiveDdFA compression problem, and this ratio is tight.
\end{theorem}

\subparagraph*{Techniques.}
As every state in a \DdFA has at most one outgoing default transition, the default transitions of an unrestricted \DdFA form a pseudoforest (a graph where each connected component contains at most one cycle). This structural property might intuitively suggest that an optimal \DdFA can be obtained by simply augmenting an MSF compression with additional default transitions. However, not every cycle of default transitions is valid. Specifically, cycles of default transitions are only valid when, for every character in the alphabet, at least one state in the cycle has an outgoing transition labeled by that character. Thus, adding the default transition that completes a cycle of default transitions may not allow for the removal of as many transitions as expected. Because of this, the optimal placement of default transitions may differ completely from the placement in the MSF compression if it allows for better cycles of default transitions.

To capture this difficulty, we design a polynomial-time reduction from the set cover problem. The core challenge is to construct an input DFA such that any equivalent optimal \DdFA directly corresponds to a solution to this set cover instance. We achieve this by creating a DFA with as many gadgets as the number of sets allowed in the set cover instance. We design these gadgets such that the choice of default transitions within a gadget corresponds to the choice of a set in the set cover instance. The key difficulty in the reduction is to ensure that the default transitions of any optimal \DdFA form a cycle going through the gadgets. By carefully placing additional transitions in the input DFA, we guarantee that maximal size reduction is attained if and only if the default transitions form a cycle that encodes a valid set cover. For a \emph{no}-instance of set cover, the default transitions may look completely different. However, in this case they allow for less reduction in size. As a result, we prove that it is at least as difficult to solve the \DdFA compression problem as the set cover problem, thereby classifying it as NP-hard.

Having established the intractability of the general problem, we next analyze the approximation guarantee of the MSF compression relative to an optimal unrestricted \DdFA. However, as we discussed, the configuration of default transitions in the two \DdFA{s} can differ substantially, complicating any direct comparison. Instead, we consider an intermediate forest-restricted \DdFA derived from the optimal \DdFA. Because the MSF compression is optimal over all forest-restricted \DdFA{s}, the size of this intermediate \DdFA provides an upper bound on the size of the MSF compression.

Although the default transitions of an optimal \DdFA may form multiple connected components, we bound the approximation ratio by analyzing the single component where the intermediate \DdFA performs worst relative to the optimum. For this component, we can characterize its structure and which transitions can be removed by default transitions. Bounding the approximation ratio on this worst-case component yields an upper bound of $3/2$ for the intermediate \DdFA, thereby establishing a $3/2$-approximation ratio for the MSF compression.

Finally, to prove that this bound is tight, we construct a family of DFAs for which the optimal \DdFA features a cycle of default transitions and two transitions for each character in the alphabet. For this family, the size of the MSF compression approaches $3/2$ times the optimum.

\subparagraph*{Outline.}
We present the problem and the corresponding prerequisites in \Cref{sec:prelim} and \Cref{sec:D2FA}. We prove \Cref{thm:set_cover_reduction} in \Cref{sec:Hardness} and \Cref{thm:apxRatio} in \Cref{sec:ACAApx}.

\begin{figure}[t]
    \centering
    \includegraphics[clip, trim=7.55cm 18.25cm 6.8cm 2.85cm, width=0.3\textwidth]{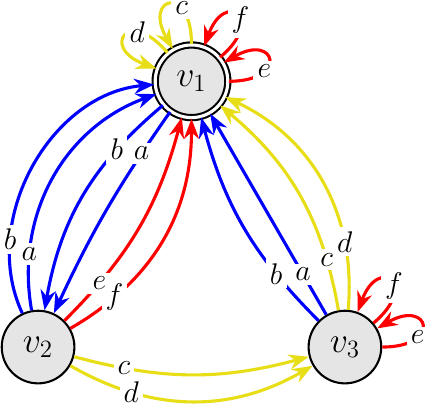}
    \hfill
    \includegraphics[clip, trim=7.55cm 18.25cm 6.8cm 2.85cm, width=0.3\textwidth]{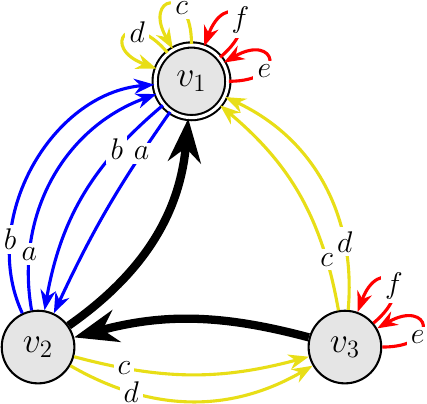}
    \hfill
    \includegraphics[clip, trim=7.55cm 18.25cm 6.8cm 2.85cm, width=0.3\textwidth]{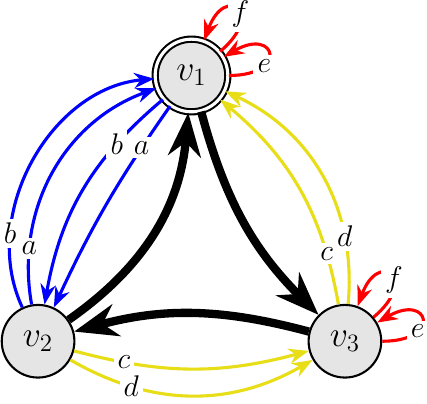}
    \caption{(Left) A DFA for the regular expression $\texttt{(C|E|A(A|E|CE}^\ast\texttt{(A|C)))}^\ast$, where $\texttt{A=(a|b)}$, $\texttt{C=(c|d)}$, and $\texttt{E=(e|f)}$ are colored blue, yellow, and red, respectively. (Middle) An equivalent MSF compression. (Right) A minimum equivalent \DdFA. This example shows that the optimal \DdFA might have cycles. In this case, the optimum has a reduction of 1 transition compared to the MSF compression.}
    \label{fig:allowingcyclesinDFA}
\end{figure}

\section{Preliminaries}\label{sec:prelim}
A \emph{deterministic finite automaton} (DFA) is a directed graph where vertices represent states and directed, labeled edges represent transitions. Formally, a DFA is defined as a $5$-tuple $D=(V, E, \Sigma, \theta, \Phi)$, where $V$ is the set of states, $E \subseteq V \times V \times \Sigma$ is the set of transitions, $\Sigma$ is the alphabet, $\theta \in V$ is the start state, and $\Phi \subseteq V$ is the set of accepting states. 

We say that the $3$-tuple $(u,v,c)\in E$ is a transition from state $u\in V$ to target state $v\in V$ with label $c\in \Sigma$. Each DFA $D$ has a corresponding \emph{transition function}, $\delta_D : V \times \Sigma \rightarrow V$, where $\delta_D(u, c) = v$ if and only if $(u, v, c) \in E$. We assume that each state in $D$ has exactly one outgoing transition for every character in $\Sigma$, ensuring that the transition function $\delta_D$ is defined on all inputs. The \emph{size} of a DFA is the number of transitions.

\section{Delayed Deterministic Finite Automata}\label{sec:D2FA}

A \emph{delayed deterministic finite automaton} (\DdFA) is a DFA that is augmented with unlabeled \emph{default transitions}. Formally, a \DdFA is a $6$-tuple $D=(V, E, \Sigma, \theta, \Phi, F)$, where $V, \Sigma, \theta,$ and $\Phi$ are defined exactly as in a standard DFA. The set $E \subseteq V \times V \times \Sigma$ contains transitions defined as in the standard DFA except that a state $u \in V$ need not have an outgoing transition for every character in $\Sigma$. The final component $F \subseteq V \times V$ is a set of default transitions. The default transitions are unlabeled and each state has at most one outgoing default transition. Each state must have a transition for every character in $\Sigma$ or a default transition to a different state. The transition function $\delta_{D}: V \times \Sigma \rightarrow V$ is defined recursively as:
\[
    \delta_{D}(u, c) =
    \begin{cases}
        v & \text{if } (u, v, c) \in E\; , \\
        \delta_{D}(v, c) & \text{if } (u, v, c) \not\in E \text{ and } (u, v) \in F\; .
    \end{cases}
\]
We require that $\delta_{D}$ is well-defined, that is $\delta_{D}(u, c)$ ends in a unique state for every $u \in V$ and $c\in \Sigma$. This implies that a \DdFA has the following property:

\begin{description}
    \item[Default cycle property.] Any cycle of default transitions has at least one outgoing transition for every character in $\Sigma$. That is, for every character $c\in \Sigma$ there is a transition $(u,v,c)\in E$, where $u$ is a state in the cycle.
\end{description}

We say that two \DdFA{}s $D= (V, E, \Sigma, \theta, \Phi, F)$ and $D' = (V', E', \Sigma', \theta', \Phi', F')$ are \emph{equivalent} if $V=V'$, $\Sigma=\Sigma'$, $\theta = \theta'$, $\Phi= \Phi'$ and 
$\delta_{D}(u, c) = \delta_{D'}(u, c)$ for all $u$ and $c$. 
The \emph{size} of a \DdFA is the number of transitions and default transitions, $|E| + |F|$. A  \DdFA is \emph{minimum} if there is no equivalent \DdFA of smaller size.

The \DdFA is a generalization of the DFA, and every DFA has a natural corresponding \DdFA  where the set of default transitions is empty. 
We say that a \DdFA is equivalent with the DFA $D = (V, E, \Sigma, \theta, \Phi)$ if it is equivalent with the corresponding \DdFA $D^2=(V, E, \Sigma, \theta, \Phi,\emptyset)$.

Given a DFA, the \emph{\CursiveDdFA compression problem} is to compute an equivalent minimum \DdFA.

\subparagraph*{Default transition graph.} The \textit{default transition graph} (DTG) of a \DdFA $D = (V,\allowbreak E,\allowbreak \Sigma,\allowbreak \theta,\allowbreak \Phi,\allowbreak F)$ is the directed subgraph $(V, F)$ containing the states and the default transitions.  Since each state in a default transition graph has at most one outgoing edge, any default transition graph is a directed pseudo-forest, that is, every weakly connected component contains at most one directed cycle.
A cycle in the DTG is referred to as a \textit{default cycle} in the \DdFA. We prove the following property. 
\begin{lemma}\label{lma:3-cycles}
    All default cycles in a minimum \CursiveDdFA  has length at least 3. 
\end{lemma}
\begin{proof}
	Let $D$ be a minimum  \CursiveDdFA. Assume for the sake of contradiction that $D$ has a default cycle of length one $(u,u)$. It follows from the default cycle property that $u$ has an outgoing transition for every character in $\Sigma$. Hence, removing the self-loop does not change the transition function and thus $D$ is not minimum. Assume now that $D$ has a default cycle of length two consisting of the default transitions $(u,v)$ and $(v,u)$.  Let $c$ be a character for which $u$ does not have a transition. By the default cycle property $v$ has a transition $(v, w, c)$. We add the transition $(u,w,c)$ and delete the transition $(v, w, c)$. We do this for every such character $c$. We then delete the default transition $(u,v)$. This produces an equivalent \DdFA with one fewer edge. Thus $D$ is not minimum. 
\end{proof}

%%%%%%%%%%%%%%%%%%%%%%%%%%%%%%%%%%% Hardness %%%%%%%%%%%%%%%%%%%%%%%%%%%%%%%%%%%

\section{NP-Hardness of the \texorpdfstring{\DdFA}{D2FA} Compression Problem}\label{sec:Hardness}
In this section, we prove \Cref{thm:set_cover_reduction}. That is, we prove that the \DdFA compression problem is NP-hard by a reduction from the set cover problem. An instance of the set cover problem consists of a tuple $(U, \mathcal{S}, k)$, where $U = \{x_1, \dots, x_u\}$ is a universe of $u$ elements, $\mathcal{S} = \{S_1, \dots, S_m\}$ is a collection of $m$ subsets of $U$, and $k \in \{2, \dots, m\}$ is a threshold. The problem is to determine whether there exists a subset $\mathcal{S}^* \subseteq \mathcal{S}$ of size at most $k$ that covers the universe (i.e., $\bigcup_{S \in \mathcal{S}^*} S = U)$.

\subsection{Construction of \texorpdfstring{$D_{SC}$}{D\_SC}}\label{sec:reduction}
Given an instance of set cover $(U, \mathcal{S}, k)$, we construct a DFA $D_{SC} = (V, E, \Sigma, \theta, \Phi)$ such that solving the \DdFA compression problem on $D_{SC}$ allows us to determine if a set cover of size at most $k$ exists. The DFA $D_{SC}$ for an instance of set cover with $k=5$ is visualized in \Cref{fig:Dsc}.

\subparagraph*{The alphabet of $D_{SC}$.}
The alphabet is defined as $\Sigma = U \cup C$, where $C$ is a set of \emph{control characters} defined as follows.  For each  $S_j \in \mathcal{S}$, we define the character sets
\[ 
    C_j^{\text{in}} = \{c_{j,1}^{\text{in}}, \dots, c_{j,|S_j|+1}^{\text{in}}\} \quad \text{and} \quad C_j^{\text{out}} = \{c_{j,1}^{\text{out}}, \dots, c_{j,|S_j|+1}^{\text{out}}\}\; .
\] 
The set of control characters is 
\[
    C = \{\$\} \cup \left(\bigcup_{j=1}^m C_j^{\text{in}}\right) \cup \left(\bigcup_{j=1}^m C_j^{\text{out}}\right)\; ,
\]
where $\$$ is a special character.
The alphabet has size $|\Sigma| = u + 1 + 2m + 2\sum_{j=1}^m |S_j|$.

\subparagraph*{The states of $D_{SC}$.}
The DFA $D_{SC}$ contains one central sink state $t$, along with $k$ identical structures we call \emph{gadgets}, denoted $G_i$ for $i = 1, \dots, k$. Each gadget $G_i$ contains $m+2$ states, labeled $\vv{i}{j}$ for $j \in \{\text{in}, \text{out}, 1, \dots, m\}$. 
The states $\vv{i}{1}, \dots, \vv{i}{m}$ are the \emph{\vijstates} and the states $\vin{i}$ and $\vout{i}$ are the \emph{\inoutstates}. 
We ignore which states are start and accepting states, as these do not influence the reductions.

\subparagraph*{The transitions of $D_{SC}$.}
The transitions of $D_{SC}$ are categorized into four types, denoted by the colors \emph{blue}, \emph{red}, \emph{pink}, and \emph{gray}. Let $S_{\text{in}} = S_{\text{out}} = \emptyset$ and $\hat{\mathcal{S}} = \{S_{\text{in}}, S_{\text{out}}, S_1, \dots, S_m\}$. For each gadget $G_i$, the transitions of $D_{SC}$ are defined as follows.
\begin{description}
    \item[Blue transitions:] For each $j \in \{\text{in}, \text{out}, 1, \dots, m\}$, we add a blue transition from $\vv{i}{j}$ to the sink state $t$ for every character in $\overline{S_j} = U \setminus S_j$.
    \item[Red transitions:] We have red transitions inside each gadget (internal) and between gadgets (external). These transitions are:
    \begin{itemize}
        \item \emph{Internal:} For each $j \in \{1, \dots, m\}$, we add a red transition from $\vin{i}$ to $\vv{i}{j}$ for every character in $C_j^{\text{in}}$, and a red transition from $\vv{i}{j}$ to $\vout{i}$ for every character in $C_j^{\text{out}}$.
        \item \emph{External:} Add a red transition to the next gadget, from state $\vout{i}$ to state $\vin{i + 1}$ with label $\$$ (with $\vout{k}$ connecting back to $\vin{1}$).
    \end{itemize}
    \item[Pink transitions:] For every red transition $(u, v, c)$, add a pink self-loop $(v, v, c)$.
    \item[Gray transitions:] To ensure $D_{SC}$ is a complete DFA, we add the following remaining transitions:
    \begin{itemize}
        \item For every state $v \neq t$, we add a gray self-loop at $v$ for any character in $\Sigma$ that does not already have a transition from $v$.
        \item For every character in $U$, we add a gray transition from $t$ to $\vin{1}$.
        \item For every character in $C$, we add a gray self-loop at $t$. 
    \end{itemize}
\end{description}
Note that there is a blue transition from  $\vin{i}$ and $\vout{i}$ to $t$ for every character in $U$, since $S_{\text{in}} = S_{\text{out}} = \emptyset$.
Also note that only $\vout{i}$ has incoming red transitions from distinct states.

Since $D_{SC}$ is completely specified with $1 + k(m+2)$ states, it has $|E| = \left(1 + k(m+2)\right)|\Sigma|$ transitions.

\begin{figure}[t]
    \centering
    \includegraphics[clip, trim=2.8cm 17.1cm 2.78cm 2.6cm, width=1.00\textwidth]{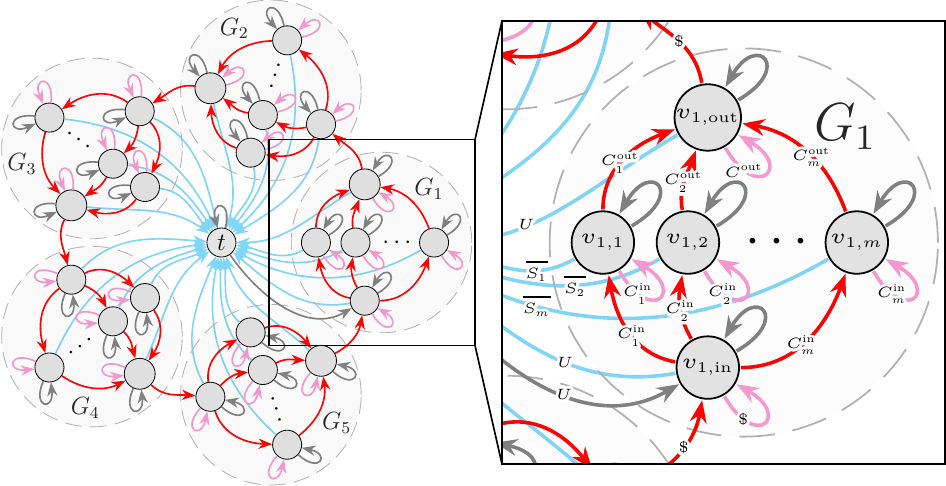}
    \caption{An illustration of the states and transitions in $D_{SC}$. The characters (or sets of characters) corresponding to each (non-gray) transition is visible in the zoomed view on the right. The remaining characters outgoing from each state corresponds to gray edges. We use the notation $C^\text{out}=\bigcup_{j=1}^m C_j^{\text{out}}$.}
    \label{fig:Dsc}
\end{figure}

\subsection{Properties of \texorpdfstring{$D_{SC}$}{D\_SC} and an Equivalent \texorpdfstring{\DdFA{}}{D2FA}}
We say that two states $u, v \in V$ in a DFA $D$ (or similarly for a \DdFA) \emph{agree} on a character $c \in \Sigma$ if $\delta_D(u, c) = \delta_D(v, c)$, i.e., their transitions with label $c$ have the same target state.
If a state $u \in V$ has an outgoing default transition $(u, v) \in F$ but no outgoing transition with label $c \in \Sigma$, we say that the transition for $c$ at $u$ is \textit{replaced} by the default transition. 

The following describes how adding default transitions to $D_{SC}$ can (or cannot) replace transitions.

\begin{lemma}\label{lem:defaultToT}
    In a \CursiveDdFA equivalent to $D_{SC}$, a default transition entering or leaving the sink state $t$ replaces no transitions.
\end{lemma}
\begin{proof}
    To show that a default transition with $t$ as an endpoint cannot replace any transitions, we must prove that $t$ does not agree with any other state on any character. We analyze this in two parts based on the alphabet partition $\Sigma = U \cup C$:
    \begin{description}
        \item[Agreement on $U$:] Every outgoing transition from $t$ labeled with a character $c \in U$ has target state $\vin{1}$. However, $\vin{1}$ has a blue transition to $t$ for every character in $U \setminus S_{\text{in}} = U$. Also, no other state has transitions labeled by $U$ with target state $\vin{1}$. Therefore, the sink state $t$ has unique transitions for every character in $U$.
        \item[Agreement on $C$:] Every outgoing transition from $t$ labeled with a character $c \in C$ is a self-loop. Because $t$ only has incoming transitions from other states labeled with characters from the universe $U$, no other state can agree with $t$ on a character in $C$.
    \end{description}
    Consequently, $t$ does not agree with any other state on any character.
\end{proof}

\begin{corollary}\label{cor:noDefaultToT}
    A minimum \DdFA equivalent to $D_{SC}$ contains no default transitions incident to the sink state $t$.
\end{corollary}

\begin{lemma}\label{lem:graytransitions}
    In a \CursiveDdFA equivalent to $D_{SC}$, no gray transition can be replaced by a default transition.
\end{lemma}
\begin{proof}
    By \Cref{lem:defaultToT}, any default transition incident to the sink state $t$ replaces no transitions; hence, it suffices to consider default transitions between non-\tstates. Assume for the sake of contradiction that a gray transition from a non-\tstate $\vv{i}{j}$ with label $c \in \Sigma$ is replaced by a default transition to another non-\tstate $v' \neq \vv{i}{j}$.
    
    All gray transitions from non-\tstates form self-loops. Thus for the default transition to replace the gray transition, $v'$ must have $\delta_{D_{SC}}(v', c) = \vv{i}{j}$. All transitions between distinct non-\tstates are red transitions, so the transition $(v',\vv{i}{j},c)$ is red. Furthermore, the target state of every red transition has a pink self-loop with matching label. Thus, \vv{i}{j} has a pink self-loop with label $c$, contradicting the assumption that the transition from \vv{i}{j} with label $c$ is gray.
\end{proof}

\begin{lemma}\label{lma:maxAgreement}
    Two distinct states in $D_{SC}$ agree on exactly $u + 1$ characters if there exists a red transition between them. Otherwise, they agree on at most $u$ characters.
\end{lemma}
\begin{proof}
    By \Cref{lem:defaultToT}, we can restrict our focus to pairs of distinct non-\tstates.
    Let $\vv{i}{j}$ and $\vv{p}{q}$ be two distinct non-\tstates. By \Cref{lem:graytransitions}, two states can only agree on characters where the corresponding two transitions (one from each state) are colored blue, red, or pink. We analyze each of the six possible color combinations of the two transitions. 
    We first prove that for four of these combinations the two states cannot agree on any character.
    \begin{itemize}
        \item[] \textbf{Red-Red:} Only the out-\inoutstate of each gadget receives incoming red transitions from multiple states. Thus, two distinct states \vv{i}{j} and \vv{p}{q} can only have the same target state via red transitions if they belong to the same gadget ($i = p$) with distinct set indices $j, q \in \{1, \dots, m\}$.
        However, this implies that their red transitions are labeled by the distinct sets $C_j^{\text{out}}$ and $C_q^{\text{out}}$.
        \item[] \textbf{Pink-Pink:} Pink transitions are self-loops, so $\delta_{D_{SC}}(\vv{i}{j}, c) = \vv{i}{j} \neq \vv{p}{q} = \delta_{D_{SC}}(\vv{p}{q}, c)$ for any character $c$.
        \item[] \textbf{Blue-Red and Blue-Pink:} All blue transitions have $t$ as target state, whereas no red or pink transitions have the $t$ as target state.
    \end{itemize}
    {\flushleft 
    This leaves only blue-blue and red-pink combinations.
    }
    \begin{itemize}
        \item[] \textbf{Blue-Blue:} All blue transitions have target state $t$. The number of characters on which $\vv{i}{j}$ and $\vv{p}{q}$ agree from blue transitions is given by $|\overline{S_j} \cap \overline{S_q}| = u - |S_j \cup S_q| \le u$.
        \item[] \textbf{Red-Pink:} There are two cases based on the presence of a red transition between the two states:
        \begin{description}
            \item[Case 1:] \emph{A red transition exists between $\vv{i}{j}$ and $\vv{p}{q}$.} By construction, pink transitions ensure that $\vv{i}{j}$ and $\vv{p}{q}$ agree on all characters used to label red transitions between them. Red transitions only occur in three configurations: (i)~from $\vin{i}$ to a \vijstate with labels $C_q^{\text{in}}$; (ii)~from a \vijstate to $\vout{i}$ with labels $C_j^{\text{out}}$; and (iii)~from $\vout{i}$ to $\vin{i+1}$ with label $\$$. The sizes of these sets are $|S_q| + 1$, $|S_j| + 1$, and $1$, respectively. This is equal to $|S_j \cup S_q| + 1$ characters in all cases as either $\vv{i}{j}$ or $\vv{p}{q}$ (or both) are boundary states with $S_{\text{in}} = S_{\text{out}} = \emptyset$.
            
            \item[Case 2:] \emph{No red transition exists between $\vv{i}{j}$ and $\vv{p}{q}$.} Because pink transitions are always self-loops, the red transitions cannot have the same target state as the pink transitions. So in this case, $\vv{i}{j}$ and $\vv{p}{q}$ do not agree on any characters from red and pink transitions.
        \end{description}
    \end{itemize}
    
    Therefore, if a red transition connects the two states, they agree on exactly $(u - |S_j \cup S_q|) + (|S_j \cup S_q| + 1) = u + 1$ characters. Conversely, if no red transition connects them, their agreement is based only on the blue-blue combination, which is at most $u$ characters.
\end{proof}

%%%%%%%%%%%%%%%%%%%%%%%%%%%%%%%%%%% Correctness of the Reduction %%%%%%%%%%%%%%%%%%%%%%%%%%%%%%%%%%%
\subsection{Correctness of the Reduction}
We now prove the following lemma from which the correctness of the reduction from set cover follows:

\begin{lemma}\label{lma:set_cover_reduction}
    The set cover instance $(U, \mathcal{S}, k)$ has a set cover of size at most $k$ if and only if there exists a \CursiveDdFA equivalent to $D_{SC}$ of size at most $|D_{SC}| - (m+2)uk$.
\end{lemma}

\begin{proof}[Proof ($\implies$ direction)]
    Suppose that $(U, \mathcal{S}, k)$ is a \textit{yes-instance} of the set cover problem. Then there exists an index sequence $Z = (z_1, \dots, z_k)$ such that the subcollection $\mathcal{S}_Z = \{S_{z_1}, \dots, S_{z_k}\} \subseteq \mathcal{S}$ covers $U$. We construct a \DdFA $D^2_{SC} = (V, E_{D^2}, \Sigma, \theta, \Phi, F)$ from $D_{SC}$ by adding a set of default transitions $F \subseteq V \times V$ and removing all transitions that can be replaced by $F$. Specifically, for each gadget $G_i$, we modify the transitions as follows:
    \begin{enumerate}
        \item[(i)] \textbf{Boundary-to-set:} Add a default transition $(\vin{i}, \vv{i}{z_i})$ and remove the transitions labeled by characters from $\overline{S_{z_i}}$ or $C_{z_i}^{\text{in}}$ from $\vin{i}$.
        
        \item[(ii)] \textbf{Set-to-boundary:} For each $j \in \{1, \dots, m\}$, add a default transition $(\vv{i}{j}, \vout{i})$ and remove the transitions labeled by characters from $\overline{S_j}$ or $C_j^{\text{out}}$ from $\vv{i}{j}$.
        
        \item[(iii)] \textbf{Inter-gadget:} Add a default transition $(\vout{i}, \vin{i+1})$ to the next gadget and remove the transitions labeled by characters from $U$ or $\$$ from $\vout{i}$ (where $\vin{k+1} = \vin{1}$).
    \end{enumerate}
    Note that only the boundary-to-set default transitions depend on the index sequence $Z$. Furthermore, while set-to-boundary default transitions replace all blue transitions from every \vijstate, each $\vin{i}$ retains its blue transitions labeled by $U\setminus \overline{S_{z_i}} = S_{z_i}$.

    For each of the $k$ gadgets, we have one boundary-to-set, $m$ set-to-boundary, and one inter-gadget default transitions for a total of $|F| = (m+2)k$ default transitions. We add default transitions only between states that have a red transition between them, allowing each default transition to replace $u + 1$ transitions (\Cref{lma:maxAgreement}). Thus, the number of replaced transitions is $(m+2)k(u+1)$. Accounting for the addition of the $|F|$ default transitions, the total size of $D^2_{SC}$ is 
    \[
        |D_{SC}| - (m+2)k(u+1) + (m+2)k = |D_{SC}| - (m+2)ku\; ,
    \]
    which matches the desired size. 

    Next, we show that $\delta_{D^2_{SC}}$ is well-defined (uniquely determines a state on every input) and that $\delta_{D^2_{SC}}=\delta_{D_{SC}}$. Since $D_{SC}$ has a well-defined transition function, and we only remove transitions that can be replaced by $F$, it suffices to show that $D_{SC}^2$ satisfies the default cycle property. By construction, $D^2_{SC}$ contains exactly one default cycle $\zeta$, traversing the states 
    \[
        \vin{1}, \vv{1}{z_1}, \vout{1}, \vin{2}, \vv{2}{z_2}, \vout{2}, \dots, \vin{k}, \vv{k}{z_k}, \vout{k}\; .
    \]
    See~\Cref{fig:D2FAFromSetCover} for an illustration.
    
    Assume for the sake of contradiction that there exists some character $c \in \Sigma$ such that the transitions labeled $c$ have been replaced at every state along $\zeta$. Consequently, all states in $\zeta$ must agree on $c$ (have the same target state). The shared target state must be the \tstate $t$, as this is the only state with incoming transitions from every gadget. Because all transitions with target state $t$ are labeled by characters from $U$, it follows that $c \in U$.
    
    While all \inoutstates have a transition to $t$ for every character in $U$, each \vijstate $\vv{i}{z_i}$ has transitions to $t$ only for characters in $\overline{S_{z_i}}$.
    Since every state in $\zeta$ agrees on $c$, we must have $c \in \overline{S_{z_i}}$ for all $i \in \{1, \dots, k\}$. This implies that $c \notin S_{z_i}$ for all $i \in \{1,\dots,k\}$, contradicting the assumption that $\mathcal{S}_Z = \{S_{z_1}, \dots, S_{z_k}\}$ covers $U$. Therefore, no such character $c$ exists.
    
    Hence, $D^2_{SC}$ is a valid \DdFA equivalent to $D_{SC}$ with a size reduction of $(m+2)ku$.
\end{proof}

\begin{proof}[Proof ($\impliedby$ direction)]
    Suppose $D_{SC}$ has an equivalent \CursiveDdFA $D^2_{SC}$ with a size reduction of at least $(m+2)ku$. By \Cref{lma:maxAgreement}, each default transition replaces at most $u+1$ transitions, giving a reduction of at most $u$. Achieving a reduction of $(m+2)ku$ thus requires at least $(m+2)k$ default transitions.
    
    Because $t$ has no incident default transitions (\Cref{cor:noDefaultToT}), each of the $(m+2)k$ non-\tstates must have exactly one outgoing default transition in $F$, each replacing $u+1$ transitions. It follows from \Cref{lma:maxAgreement} that each such default transition connects states that have a red transition between them.
    
    Since we have $(m+2)k$ default transitions between $(m+2)k$ states (and an out-degree of $1$), $F$ contains at least one default cycle $\zeta$. By \Cref{lma:3-cycles}, $\zeta$ has length at least $3$. 
    Because the default transitions follow red transitions, and any path composed solely of \inoutstates has length at most $2$, $\zeta$ must contain at least one \vijstate. Let $Z$ be the index sequence $(z_1,\dots,z_\ell)$ such that $\zeta$ contains the sequence of set-states $\vv{i_1}{z_1}, \vv{i_2}{z_2}, \dots, \vv{i_\ell}{z_\ell}$. We collect these indices into the subcollection $\mathcal{S}_Z = \{S_{z_1}, S_{z_2}, \dots, S_{z_\ell}\} \subseteq \mathcal{S}$. We claim that $\mathcal{S}_Z$ is a valid set cover for $U$ of size at most $k$.
    
    First, to see that $|\mathcal{S}_Z| \le k$, recall that default transitions achieving maximal savings are between states that have a red transition between them (between a \vijstate $\vv{i}{j}$ and a \inoutstate ($\vin{i}$ or $\vout{i}$) within the same gadget, or between \inoutstates of adjacent gadgets ($\vout{i}$ and $\vin{i+1}$)). Consequently, if $\zeta$ contains two \vijstates from the same gadget, it must form a 4-cycle confined entirely to that gadget, in which case $|\mathcal{S}_Z| = 2 \leq k$ (since $k \ge 2$). Otherwise, $\zeta$ visits at most one \vijstate per gadget, so $|\mathcal{S}_Z| \leq k$.
    
    Second, we show that $\mathcal{S}_Z$ covers $U$. Assume for contradiction that some element $c \in U$ is not covered by $\mathcal{S}_Z$. Then $c \in \overline{S_{z_r}}$ for all \vijstates $\vv{i_r}{z_r} \in\zeta$, where $r \in \{1, \dots, l\}$. Consequently, every \vijstate $\vv{i_r}{z_r}$ has a transition to $t$ with label $c$. Because every \inoutstate in $\zeta$ also has a transition to $t$ with label $c$, the transition for $c$ has been replaced for all states in $\zeta$. Consequently, $D^2_{SC}$ does not possess the default cycle property, contradicting its validity.
    
    Hence, $\mathcal{S}_Z$ covers $U$, and $(U, \mathcal{S}, k)$ is a \textit{yes-instance} of set cover.
\end{proof}

The proof of \Cref{thm:set_cover_reduction} follows from \Cref{lma:set_cover_reduction}.

\begin{figure}[t]
    \centering
    \includegraphics[clip, trim=2.8cm 17.1cm 2.78cm 2.6cm, width=1.00\textwidth]{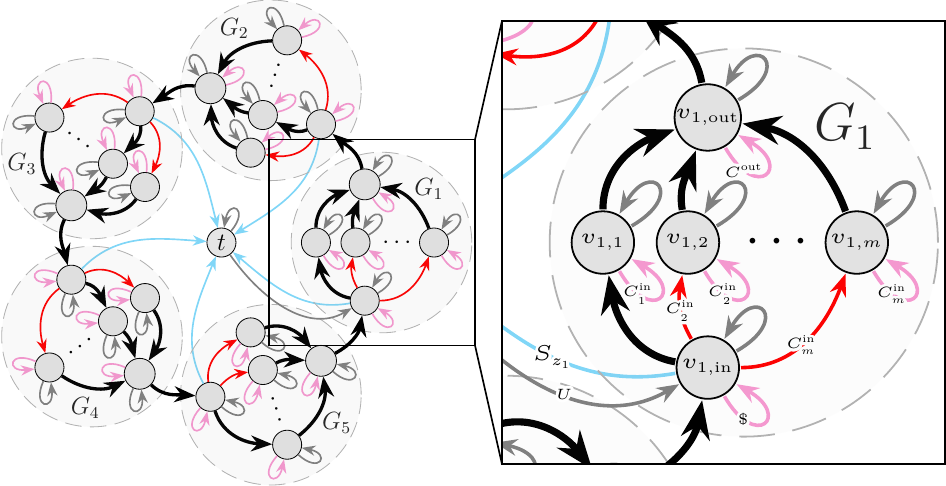}
    \caption{The \DdFA $D^2_{SC}$ (equivalent to the DFA from \Cref{fig:Dsc}). The black, bold transitions represent the added default transitions. The default transitions from the states $\vin{1}, \dots, \vin{k}$ specify the corresponding solution to the set cover instance.}
    \label{fig:D2FAFromSetCover}
\end{figure}

%%%%%%%%%%%%%%%%%%%%%%%%%%%%%%%%%%% MST approximation %%%%%%%%%%%%%%%%%%%%%%%%%%%%%%%%%%%
\section{Approximation Ratio of the MSF Compression}\label{sec:ACAApx}
In this section, we establish \Cref{thm:apxRatio}. That is, we prove that using the MSF compression for the \DdFA compression problem has an approximation ratio of $3/2$. We also show that this is tight.

\paragraph*{Upper bound on the approximation ratio}
The key idea is to transform an optimal solution into a forest-restricted \DdFA, allowing us to compare this intermediate \DdFA against both the MSF compression and the optimum.

Specifically, let $D$ be an input DFA with alphabet $\Sigma$ of size $\sigma$ that maximizes the approximation ratio $\alpha = {|D^2_{\text{MSF}}|}/{|D^2_{\text{OPT}}|}$. Here, $D^2_{\text{MSF}}$ is the MSF compression and $D^2_{\text{OPT}}$ is an optimal \DdFA, both equivalent to $D$.

From $D^2_{\text{OPT}}$, we construct an acyclic \DdFA, denoted $D^2_{\text{O2A}}$, as follows. For each cycle in the default transition graph (DTG) of $D^2_{\text{OPT}}$, we remove the default transition that replaces the fewest transitions, and insert the transitions it replaced.

% From the optimal $D^2_{\text{OPT}}$, we construct an acyclic \DdFA, denoted $D^2_{\text{O2A}}$, by identifying each cycle in the default transition graph (DTG) of $D^2_{\text{OPT}}$, removing the single default transition that yields the smallest size reduction, and inserting the transitions replaced by that default transition.

% \inge{is this well-defined? Is it the one where we have to reinsert the fewest transitions that was replaced by this?}
% \inge{"... as follows". For each cycle in the default transition graph (DTG) of $D^2_{\text{OPT}}$, we find the default transition in the cycle that replaces the fewest transitions, remove that default transition, and insert the transitions it replaced.}

We decompose $D^2_{\text{OPT}}$ and $D^2_{\text{O2A}}$ into subgraphs based on the weakly connected components in the DTG of $D^2_{\text{OPT}}$. Specifically, let $V_i$ be the states of the $i$th weakly connected component in the DTG of $D^2_{\text{OPT}}$. The subgraph $C^2_{i, \text{OPT}}$ includes all transitions and default transitions in $D^2_{\text{OPT}}$ outgoing from states in $V_i$ and the endstates of these transitions. Correspondingly, the subgraph $C^2_{i, \text{O2A}}$ includes all transitions and default transitions in $D^2_{\text{O2A}}$ outgoing from states in $V_i$ and the endstates of these transitions. Note that every transition and default transition of $D^2_{\text{OPT}}$ and $D^2_{\text{O2A}}$ is in exactly one component.

% \inge{Specifically, let $V_i$ be the states of the $i$th connected component of the DTG of $D^2_{\text{OPT}}$. The subgraph $C^2_{i, \text{OPT}}$ includes all transitions and default transitions in $D^2_{\text{OPT}}$ from states in $V_i$ and the endstates of these transitions. Note that, any transition and default transition of  $D^2_{\text{O2A}}$ is in exactly one component.}

% Specifically, let $C^2_{i, \text{OPT}}$ denote the $i$-th subgraph of $D^2_{\text{OPT}}$. The subgraph includes all transitions and default transitions in $D^2_{\text{OPT}}$ from states in the $i$-th connected component of the DTG of $D^2_{\text{OPT}}$. Correspondingly, let $C^2_{i, \text{O2A}}$ denote the $i$-th subgraph of $D^2_{\text{O2A}}$.
Let $C^2_{\text{OPT}}$ and $C^2_{\text{O2A}}$ be the component pair maximizing the ratio ${|C^2_{i, \text{O2A}}|} /{|C^2_{i, \text{OPT}}|}$, referred to as the worst-case component ratio $\beta$, such that
\[
    \beta = \frac{|C^2_{\text{O2A}}|}{|C^2_{\text{OPT}}|} = \max_i\left\{\frac{|C^2_{i, \text{O2A}}|}{|C^2_{i, \text{OPT}}|}\right\}\; .
\]
Since $|D^2_{\text{O2A}}| = \sum_i |C^2_{i, \text{O2A}}|$ and $|D^2_{\text{OPT}}| = \sum_i |C^2_{i, \text{OPT}}|$, and $|C^2_{i, \text{O2A}}| \le \beta \cdot |C^2_{i, \text{OPT}}|$ for all $i$, we have
\begin{equation}\label{eq:betaGeO2A}
    \frac{|D^2_{\text{O2A}}|}{|D^2_{\text{OPT}}|} = \frac{\sum_i |C^2_{i, \text{O2A}}|}{\sum_i |C^2_{i, \text{OPT}}|} \leq \frac{\sum_i \beta \cdot |C^2_{i, \text{OPT}}|}{\sum_i |C^2_{i, \text{OPT}}|} = \beta.
\end{equation}
    
Thus, the worst-case component ratio $\beta$ is an upper bound to the approximation ratio of $D^2_{\text{O2A}}$. In order to bound $\beta$, the following lemma establishes the structure of $C^2_{\text{OPT}}$ when $\beta$ is maximized.
\begin{lemma}\label{lma:worstCaseStructure}
    % If the worst-case component ratio is maximized, then $C^2_{\text{OPT}}$ has at least $2\sigma$ transitions, and its default transitions form a simple cycle.
    % Let $C^2_{\text{OPT}}$ and $C^2_{\text{O2A}}$ be a component pair achieving the worst-case component ratio $\beta$. If $\beta$ is maximized, then $C^2_{\text{OPT}}$ has at least $2\sigma$ transitions, and its default transitions form a simple cycle.
    Let $C^2_{\text{OPT}}$ and $C^2_{\text{O2A}}$ be a component pair achieving the component ratio $\beta$. If $\beta$ is maximized over all valid input DFAs, then $C^2_{\text{OPT}}$ has at least $2\sigma$ transitions, and its default transitions form a simple cycle.
\end{lemma}
\begin{proof}
    If the DTG of $C^2_{\text{OPT}}$ contained no cycles, then $C^2_{\text{OPT}} = C^2_{\text{O2A}}$ and $\beta = 1$, which is not maximal. Therefore, the DTG of $C^2_{\text{OPT}}$ contains a cycle. Furthermore, since the DTG is a pseudo-forest, it contains exactly one cycle.
    
    Any default transition not part of the default cycle in $C^2_{\text{OPT}}$ is also present in $C^2_{\text{O2A}}$. Thus, removing the states of $D$ (and their outgoing transitions) that are not part of the default cycle in the DTG of $C^2_{\text{OPT}}$, subtracts an equal amount from $|C^2_{\text{OPT}}|$ and $|C^2_{\text{O2A}}|$. This increases the component ratio, contradicting the maximality of $\beta$. Consequently, all default transitions in $C^2_{\text{OPT}}$ are part the default cycle, implying that its DTG is a simple cycle.

    We now show that $C^2_{\text{OPT}}$ has at least $2\sigma$ transitions.
    Suppose for contradiction $C^2_{\text{OPT}}$ contains fewer than $2\sigma$ transitions. By the pigeonhole principle and the default cycle property, at least one character $c \in \Sigma$ appears as the label of a transition exactly once in $C^2_{\text{OPT}}$. Consequently, all states in $C^2_{\text{OPT}}$ must agree on character $c$, which implies that $C^2_{\text{O2A}}$ also contains only one transition labeled $c$. Removing this label from $\Sigma$ and corresponding transition yields a component with ratio $\frac{|C^2_{\text{O2A}}| - 1}{|C^2_{\text{OPT}}| - 1} > \beta$, contradicting the maximality of $\beta$. Hence, $C^2_{\text{OPT}}$ contains at least $2\sigma$ outgoing transitions.
\end{proof}
We now derive the upper bound on the worst-case component ratio $\beta$.
\begin{lemma}\label{lma:componentRatio}
    The worst-case component ratio is upper bounded by
    \[
        \frac{n' + \sigma\left(3 - \frac{2}{n'}\right) - 1}{n' + 2\sigma}\; ,
    \]
    where $n'$ is the number of states in $C^2_{\text{OPT}}$.
\end{lemma}
\begin{proof}
    By \Cref{lma:worstCaseStructure}, $C^2_{\text{OPT}}$ has at least $2\sigma$ transitions out of the $n'\sigma$ transitions in $D$. Its default transitions therefore replace at most $n'\sigma - 2\sigma = \sigma(n' - 2)$ transitions in total.
    
    Let $x$ be the minimum number of transitions replaced by any default transition in $C^2_{\text{OPT}}$. It follows that $|C^2_{\text{O2A}}| = |C^2_{\text{OPT}}| + x - 1$. The minimum is no larger than the average, yielding:
    \[
        x \le \frac{\sigma(n' - 2)}{n'} = \sigma\left(1-\frac{2}{n'}\right)\; .
    \]
    Since $|C^2_{\text{OPT}}| \ge n' + 2\sigma$ (\Cref{lma:worstCaseStructure}), substituting $x$ in the component ratio yields
    \[
        \beta = \frac{|C^2_{\text{OPT}}| + x - 1}{|C^2_{\text{OPT}}|} 
                \le \frac{n' + 2\sigma + \sigma(1-\frac{2}{n'}) - 1}{n' + 2\sigma} 
                = \frac{n' + \sigma\left(3 - \frac{2}{n'}\right) - 1}{n' + 2\sigma}\; . \qedhere
    \]
\end{proof}

Finally, using the upper bound on $\beta$, we bound the approximation ratio of the MSF compression.
\begin{lemma}\label{lma:apxUpperBound}
    The approximation ratio of the MSF compression is $3/2$. 
\end{lemma}
\begin{proof}
    Because $D^2_{\text{O2A}}$ is acyclic and the MSF compression is optimal over all acyclic \DdFA{s}~\cite{IntroOfD2FA}, we have 
    \[
        \alpha = \frac{|D^2_{\text{MSF}}|}{|D^2_{\text{OPT}}|} \le \frac{|D^2_{\text{O2A}}|}{|D^2_{\text{OPT}}|} \le \beta\; ,
    \]
    where the last inequality comes from \Cref{eq:betaGeO2A}. Using \Cref{lma:componentRatio} we get
    \[
        \alpha  \le \beta
                \le \frac{n' + \sigma\left(3 - \frac{2}{n'}\right) - 1}{n' + 2\sigma}
                = 1 + \frac{\sigma\left(1 - \frac{2}{n'}\right) - 1}{n' + 2\sigma} 
                \le 1 + \frac{\sigma}{n' + 2\sigma} 
                \le \frac{3}{2}\; . \qedhere
    \]
\end{proof}

% ######### LOWER BOUND ##############
\paragraph*{Lower bound on the approximation ratio}
We establish a matching lower bound on the approximation ratio $\alpha$ of the MSF compression. We first introduce a family of DFAs for which an optimal \DdFA is known, and then show that this family can be used to obtain an approximation ratio arbitrarily close to $3/2$.

\begin{figure}[t]
    \centering
    \includegraphics[clip, trim=4.65cm 12.5cm 4.55cm 2.8cm, width=0.45\textwidth]{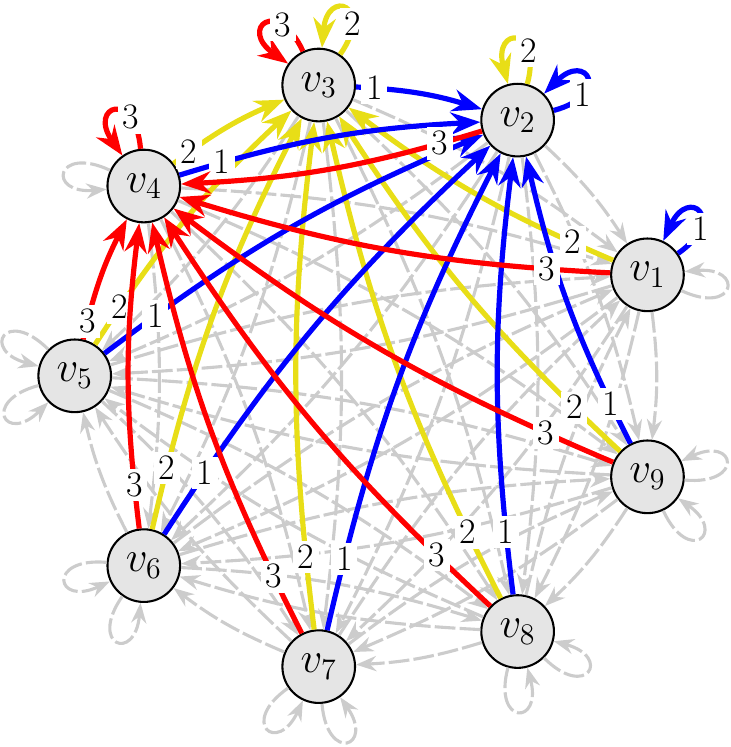}
    \hfill
    \includegraphics[clip, trim=4.65cm 12.5cm 4.55cm 2.8cm, width=0.45\textwidth]{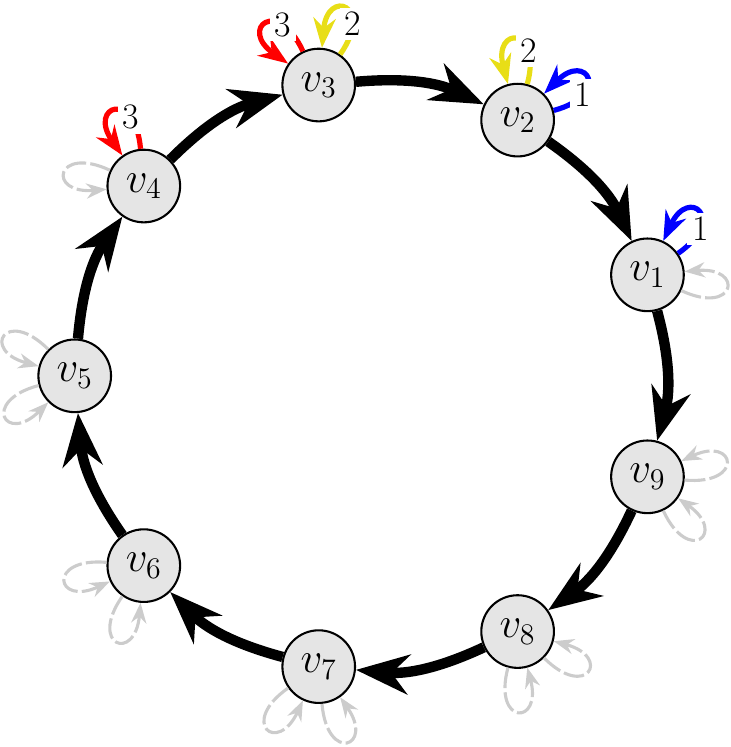}
    \caption{An example of our lower bound construction for $n=9$ and the corresponding minimum equivalent \DdFA. The sets of characters $\Sigma_1, \Sigma_2, \Sigma_3$ have their corresponding transitions labeled by $1,2,3$, and are colored blue, yellow, and red, respectively. The remaining transitions are gray and dashed.}
    \label{fig:ApxLowerBound}
\end{figure}

\begin{lemma}\label{lma:apxLowerBound}
    For each $n > 3$ and $k\in\mathbb{N}$, there is a DFA with $n$ states and an alphabet of size $\sigma=kn$, where the MSF compression is larger than the minimum equivalent \CursiveDdFA by a factor of 
    \[
    \frac{n + \sigma\left(3 - \frac{2}{n}\right) - 1}{n + 2\sigma}\; .
    \]
\end{lemma}

\begin{proof}
    Let $D$ be a DFA with $n$ states and an alphabet of size $\sigma=kn$. Partition the alphabet $\Sigma$ into $n$ sets of characters, $\Sigma_1, \dots, \Sigma_n$, each of size $k$. For each $i\in\{1,\dots, n\}$, transitions labeled by the characters from $\Sigma_i$ go to state $v_{i+1}$ (with index modulo $n$) from every state except $v_i$, where the transitions are self-loops. 
    
    The size of the DFA is $|D|=n\sigma$. Note that any two states $v_i$ and $v_j$ agree on exactly $(n-2)\sigma/n$ characters, namely all characters except those in the sets $\Sigma_i$ or $\Sigma_j$. 
    
    We construct an optimal \DdFA $D^2_{\text{OPT}}$ by adding a default transition from each $v_i$ to $v_{i-1}$ (where $v_1$ has a default transition to $v_n$). The DFA and the optimal \DdFA is visualized in \Cref{fig:ApxLowerBound}. Each default transition replaces $(n-2)\sigma/n$ transitions so the total size of the \DdFA is $|D^2_{\text{OPT}}| = |D| - n(n - 2)\sigma/n + n
                = n\sigma-(n-2)\sigma+n
                =2\sigma+n$.
    %\begin{align*}
       % %|D^2_{\text{OPT}}| &= |D| - n(n - 2)\sigma/n + n\\
          %      &= %n\sigma-(n-2)\sigma+n\\
                %&=2\sigma+n\; .
   % \end{align*}
   %
    As there are no characters that all states agree on, this \DdFA possesses the default cycle property. 
    
    Now consider the MSF compression, $D_\text{MSF}^2$. As the default transitions of this \DdFA do not create cycles, it has at most, and thus exactly, $n-1$ default transitions, each replacing $(n-2)\sigma/n$ transitions. The size of this \DdFA is
    \begin{align*}
        |D^2_{\text{MSF}}| &= |D^2_{\text{OPT}}| + (n-2)\sigma/n-1\\
                    &=n+2\sigma+(n-2)\sigma/n-1\\
                    &=n + \sigma\left(3 - \frac{2}{n}\right) - 1\; .
    \end{align*}
    Taking the ratio of their sizes gives the desired bound:
    \[
        \frac{|D^2_{\text{MSF}}|}{|D^2_{\text{OPT}}|} = \frac{n + \sigma\left(3 - \frac{2}{n}\right) - 1}{2\sigma+n}\; .\qedhere
    \]
\end{proof}

We now show that the family of DFAs from \Cref{lma:apxLowerBound} contains DFAs for which the MSF compression has an approximation ratio arbitrarily close to $3/2$.

\begin{lemma}\label{lma:apxLowerBoundTight}
    For each $\epsilon > 0$, there exists an instance of the \CursiveDdFA compression problem where the MSF compression has an approximation ratio of at least $3/2 - \epsilon$.
\end{lemma}

\begin{proof}
    Consider the DFA from \Cref{lma:apxLowerBound} with $n$ states and an alphabet of size $\sigma=kn$. The MSF compression of this DFA is larger than the optimal \DdFA by a factor of 
    \[
        \frac{n+kn(3-\frac{2}{n})-1}{n+2kn} 
        = \frac{(3kn+\frac{3}{2}n)-\frac{1}{2}n-(2k+1)}{2kn+n}
        = \frac{3}{2}-\frac{1}{4k+2}-\frac{1}{n}\; .
    \]
    
    For any given $\epsilon > 0$, we choose $k$ and $n$ sufficiently large so that $\frac{1}{4k + 2} + \frac{1}{n} \le \epsilon$, completing the proof.
\end{proof}

\Cref{lma:apxUpperBound} together with \Cref{lma:apxLowerBoundTight} proves \Cref{thm:apxRatio}.

\paragraph*{Disclosure of AI Usage}
Generative AI tools (specifically, free versions of ChatGPT/Gemini) were used during the preparation of this paper solely to edit language, fix grammar, and improve readability. A GPT-5.6 sol model was used to help find the example used in \Cref{fig:allowingcyclesinDFA}. All ideas, proofs, and technical results were developed entirely by the authors. The authors have checked and verified all text and take full responsibility for the contents of this paper.

\newpage
\bibliographystyle{abbrv}
\bibliography{bibliography}
\end{document}